\documentclass[preprint,12pt,authoryear]{elsarticle}

\usepackage{amssymb}
\usepackage{amsthm}

\usepackage{multirow}
\usepackage{multicol}
\usepackage[german,english]{babel}
\usepackage[scaled]{helvet}
\usepackage{bm}
\usepackage{amsmath}
\usepackage{mathtools}
\usepackage{setspace}
\usepackage{enumitem}
\usepackage{booktabs}
\usepackage{float}
\usepackage[dvipsnames]{xcolor}
\usepackage{tikz}
\usetikzlibrary{calc, decorations.pathreplacing, arrows.meta, positioning}

\newtheorem{proposition}{Proposition}[section]
\newtheorem{lemma}[proposition]{Lemma}

\makeatletter
\newcommand*\bigcdot{\mathpalette\bigcdot@{.5}}
\newcommand*\bigcdot@[2]{\mathbin{\vcenter{\hbox{\scalebox{#2}{$\m@th#1\bullet$}}}}}
\makeatother

\journal{arXiv.org}

\begin{document}

\begin{frontmatter}



\title{Hidden semi-Markov models with inhomogeneous state dwell-time distributions}


\author[inst1]{Jan-Ole Koslik}

\affiliation[inst1]{organization={Bielefeld University},
            addressline={Universitätsstraße 25}, 
            city={Bielefeld},
            postcode={33615}, 
            country={Germany}}



\begin{abstract}
The well-established methodology for the estimation of hidden semi-Markov models (HSMMs) as hidden Markov models (HMMs) with extended state spaces is further developed to incorporate covariate influences across all aspects of the state process model, in particular, regarding the distributions governing the state dwell time. The special case of periodically varying covariate effects on the state dwell-time distributions --- and possibly the conditional transition probabilities --- is examined in detail to derive important properties of such models, namely the periodically varying unconditional state distribution as well as the overall state dwell-time distribution. Through simulation studies, we ascertain key properties of these models and develop recommendations for hyperparameter settings. Furthermore, we provide a case study involving an HSMM with periodically varying dwell-time distributions to analyse the movement trajectory of an arctic muskox, demonstrating the practical relevance of the developed methodology.
\end{abstract}




\end{frontmatter}



\section{Introduction}
\label{sec1}

Hidden Markov models (HMMs) have emerged as popular tools for analysing a wide range of time series data across domains like ecology \citep{mcclintock2020uncovering, farhadinia2020understanding}, economics \citep{rahman2020financial, chen2022renewable}, climatology \citep{stoner2020advanced, nyongesa2020non}, medicine \citep{kwon2020dpvis, soper2020hidden} and finance \citep{liu2012stock}. Such models comprise two stochastic processes in discrete time, namely an unobserved Markovian state process and an observed state-dependent process. Observations are assumed to be generated by one of $N$ possible distributions, which, at any time point, is chosen by the underlying state process. For an exhaustive description, see \cite{zucchini2016hidden}. HMMs provide a powerful tool for time series data with underlying correlated states, as the model's hidden states can often be treated as interpretable entities, allowing for better understanding of otherwise very complex data. For instance, in animal movement analysis, hidden states can be treated as proxies for distinct behavioural patterns \citep{van2019classifying} while in econometric applications, they can be interpreted as switching market regimes \citep{guidolin2011markov}. 

However, in HMMs governed by a homogeneous first-order Markovian state process, the time spent in a hidden state, also called the state dwell time or sojourn time, necessarily follows a geometric distribution. This distribution is characterised by being \textit{memoryless} and having a \textit{monotonously decreasing} probability mass function. This implies, {\it inter alia}, that the most likely duration of a stay in a state is one sampling unit, which can be highly unrealistic in practical applications, for example when trying to capture the stochastic dynamics of the resting state of humans or animals with regular sleeping patterns.

To directly mitigate these limitations, so-called hidden semi-Markov models (HSMMs) have emerged as a flexible and powerful generalisation of HMMs, allowing for the explicit specification and estimation of state dwell-time distributions. Such models have attracted attention across disciplines like biology \citep{guedon2003estimating, guedon2005hidden}, finance \citep{bulla2006stylized}, ecology \citep{langrock2011hidden, van2015hidden} and environmental science \citep{rojas2021bayesian}. Despite their promise to generalise HMMs, the commonly applied HSMMs are limited to time-homogeneous dwell-time distributions, which may again be an unrealistic modelling assumption. Furthermore, in fields like ecology, it has become common practice to include inhomogeneity in the state process of HMMs by linking the transition probabilities to external covariates \citep{papastamatiou2018activity}.

Recently, \citet{koslik2023inference} showed that the state dwell-time distributions of Markov chains with periodically varying transition probabilities can in fact deviate substantially from geometric distributions. At first glance, this result can be seen as an argument in favour of using simple Markovian (instead of the more complex semi-Markovian) state processes, as the common objection concerning the implied geometric state dwell-time distributions to some extent does not apply anymore once covariates such as time of day are included in the model. However, such covariates alone can not be expected to always fully capture the stochastic state dynamics, as for example the stays in a resting state will typically depend on both a) the day-night cycle but also b) the actual time of initiation of the resting mode --- a typical human being could for example sleep from 11 pm to 7 am on an average night, but a deviation from this pattern to an earlier bedtime would likely also shift the time of them getting up in the morning. The natural translation of such a pattern into a mathematical model would involve temporal covariates for addressing a), and additionally semi-Markovian structures for capturing b).


To fill this gap and thereby extend the versatility of the HSMM framework, we propose a novel framework that accommodates inhomogeneous dwell-time distributions varying with external covariates (e.g.\ time of day, but also other covariates). In this contribution, we i) generalise the model formulation and inference procedures of HSMMs to accommodate such inhomogeneities and ii) derive important properties of such models in the special case of periodically varying covariates such as time of day or day of year. Through simulation experiments and the application of our model to movement data of an arctic muskox, we demonstrate the efficacy and practical implications of this approach.

The paper proceeds as follows: Section \ref{sec:Methodology} explains the basic model structure and inference for homogeneous HSMMs followed by modifications necessary to accommodate inhomogeneous dwell-time distributions. In Section \ref{sec:simulation} we conduct simulation experiments to establish important properties of the maximum likelihood estimator for inhomogeneous HSMMs and develop recommendations for the specification of hyperparameters. A case study, applying the novel model class to a movement data set of an arctic muskox is given in Section \ref{sec:application}.

\section{Methodology}
\label{sec:Methodology}

\subsection{Basic model formulation}

An HSMM is a doubly stochastic process comprising a latent process $\{C_t\}$ and an observed process $\{X_t\}$. In contrast to \textit{hidden Markov models}, the latent process is an $N$-state \textit{semi-Markov chain}, a generalisation of a Markov chain to arbitrary state dwell-time distributions. At any given time-point $t$, conditional on the semi-Markov process being in a state $i$ of the state space $\mathcal{S} = \{1, \dotsc, N\}$, $X_t$ is independent of $C_k$ and $X_k$ for all $k \neq t$ and thereby generated by a state-dependent distribution $f(x_t \mid C_t = i)$ which we denote by $f_i(x_t)$ for ease of notation.

The latent semi-Markov chain can be characterised by two components: First, a set of probability distributions on the positive integers that determine the time spent in each hidden state, defined by probability mass functions 
$$
d_i(r) = \Pr(C_{t+r} \neq i, C_{t+r-1} = i, \dotsc, C_{t+1} = i \mid C_t = i, C_{t-1} \neq i).
$$
for each duration $r \in \mathbb{N}$ and each state $i \in \mathcal{S}$. Second, the transition dynamics, given that a state is left, are governed by a so-called embedded Markov chain $\{S_k\}$ with conditional transition probability matrix (t.p.m.) $\bm{\Omega}$, with entries
$$
\omega_{ij} = \Pr(S_{k+1} = j \mid S_k = i, S_{k+1} \neq i)
$$
and $\omega_{ii} = 0$, i.e.\ the embedded chain cannot stay in a state for longer than one time step. In general, it needs not to be homogeneous and the conditional transition probabilities can be modelled as functions of external covariates \citep[see][]{langrock2014modeling}, however, we initially assume homogeneity for ease of notation. 

Now consider the process $\{C_t\}$ generated by the embedded chain and the dwell-time distributions $d_i, \: i \in \mathcal{S}$ as follows: each event $S_k = i$ generates $r_k$ values all equal to $i$, where $r_k$ is a realisation of the according dwell-time distribution $d_{i}$, as depicted in Figure \ref{fig: structure semi-Markov chain}. We then call the generated process $\{C_t\}$ a semi-Markov process, as generally, it is not Markovian. Setting all $d_i$ to be geometric distributions again yields the special case of a Markov process. Thus, homogeneous semi-Markov processes constitute a flexible generalisation of homogeneous Markov chains. Figure~\ref{fig: structure semi-Markov chain} illustrates the generation of a semi-Markov process.

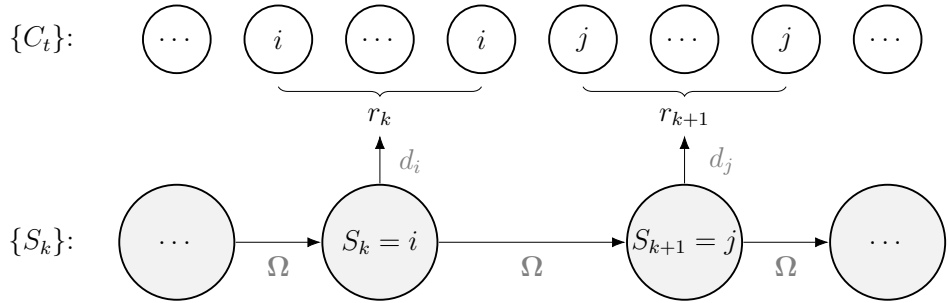
\begin{figure}
    \centering
    \scalebox{0.9}{
    \begin{tikzpicture}
        \coordinate (A) at (0,0);
        \coordinate (B) at (3,0);
        \coordinate (C) at (7.5,0);
        \coordinate (D) at (10.5,0);
        
        \filldraw[fill=black!5, thick] (A) circle (0.85);
        \draw (A) node {$\dotsc$};
        \filldraw[fill=black!5, thick] (B) circle (0.85);
        \draw (B) node {$S_{k}=i$};
        \filldraw[fill=black!5, thick] (C) circle (0.85);
        \draw (C) node {$S_{k+1}=j$};
        \filldraw[fill=black!5, thick] (D) circle (0.85);
        \draw (D) node {$\dotsc$};

        \draw[-{Latex[length=2mm]}] ($(A)+(0.87,0)$) -- ($(B)-(0.87,0)$);
        \draw[-{Latex[length=2mm]}] ($(B)+(0.87,0)$) -- ($(C)-(0.87,0)$);
        \draw[-{Latex[length=2mm]}] ($(C)+(0.87,0)$) -- ($(D)-(0.87,0)$);

        \draw[thick] ($(A)+(0,3)$) circle (0.5);
        \draw ($(A)+(0,3)$) node {$\dotsc$};
        \draw[thick] ($(B)+(-1.5,3)$) circle (0.5);
        \draw ($(B)+(-1.5,3)$) node {$i$};
        \draw[thick] ($(B)+(0,3)$) circle (0.5);
        \draw ($(B)+(0,3)$) node {$\dotsc$};
        \draw[thick] ($(B)+(1.5,3)$) circle (0.5);
        \draw ($(B)+(1.5,3)$) node {$i$};

        \draw[thick] ($(C)+(-1.5,3)$) circle (0.5);
        \draw ($(C)+(-1.5,3)$) node {$j$};
        \draw[thick] ($(C)+(0,3)$) circle (0.5);
        \draw ($(C)+(0,3)$) node {$\dotsc$};
        \draw[thick] ($(C)+(1.5,3)$) circle (0.5);
        \draw ($(C)+(1.5,3)$) node {$j$};
        \draw[thick] ($(D)+(0,3)$) circle (0.5);
        \draw ($(D)+(0,3)$) node {$\dotsc$};

        \draw[-{Latex[length=2mm]}] ($(B)+(0,0.87)$) -- ($(B)+(0,2)-(0,0.42)$);
        \draw[-{Latex[length=2mm]}] ($(C)+(0,0.87)$) -- ($(C)+(0,2)-(0,0.42)$);
        \draw[color=gray] ($(B)+(0.45,1.2)$) node {$d_i$};
        \draw[color=gray] ($(C)+(0.55,1.2)$) node {$d_j$};
        \draw[decorate, decoration={brace, amplitude = 4pt}] ($(B)+(1.5,2.3)$) -- ($(B)+(-1.5,2.3)$) node[midway, below, yshift = -5pt]{$r_k$};
        \draw[decorate, decoration={brace, amplitude = 4pt}] ($(C)+(1.5,2.3)$) -- ($(C)+(-1.5,2.3)$) node[midway, below, yshift = -5pt]{$r_{k+1}$};

        \draw ($(A) + (-2,0)$) node {$\{S_k\}$:};
        \draw ($(A) + (-2,3)$) node {$\{C_t\}$:};

        \draw[color=gray] ($(A)+(1.5,-0.35)$) node {$\bm{\Omega}$}; 
        \draw[color=gray] ($(B)+(2.25,-0.35)$) node {$\bm{\Omega}$}; 
        \draw[color=gray] ($(C)+(1.5,-0.35)$) node {$\bm{\Omega}$}; 
    \end{tikzpicture}
    }
    \caption{Visualisation of how the semi-Markov process $\{C_t\}$ is generated by the embedded Markov chain and the state dwell-time distributions $d_i$.}
    \label{fig: structure semi-Markov chain}
\end{figure}

\subsection{Parameter estimation and inference tools}

The difficulty in conducting inference for HSMMs lies in the relaxation of the Markov property which hinders the use of well-established recursive inference procedures.
Here, we build on the approach presented by \citet{langrock2011hidden} where the $N$-state semi-Markov process $\{C_t\}$ is approximated by a Markov chain $\{C_t^*\}$ operating on an enlarged state space.
The basic idea is to define a Markov chain with a large state space, with the transitions among the states structured in a very specific way so as to imply stochastic dynamics that very closely approximate the properties of the HSMM. The advantage of this approach is that it enables to apply the powerful HMM machinery also to the more complex class of HSMMs.

Specifically, the strategy involves representing each state $i \in \mathcal{S}$ of the semi-Markov process $\{C_t\}$ by a state aggregate of $N_i$ states defined as
$$
I_i = \Bigl\{n: \sum_{k=0}^{i-1}N_k < n \leq \sum_{k=0}^i N_k \Bigr\}.
$$
Additionally, we define $I_i^- = \mathrm{min}(I_i)$ and $I_i^+ = \mathrm{max}(I_i)$ to be the lowest and highest states in each state aggregate $I_i$. In an HSMM with the semi-Markov process represented by such state aggregates, each state in the state aggregate $I_i$ gives rise to the same state-dependent distribution. Hence, on an observation level, one cannot distinguish between states belonging to the same state aggregate and the aggregates are only used to approximate the dwell-time structure of the semi-Markov chain.

To allow the state aggregates to represent the dwell-time distributions $d_i, \: i \in \mathcal{S}$, as well as the conditional t.p.m.\ $\bm{\Omega}$, the t.p.m.\ of the approximating Markov chain $\{C_t^*\}$ on the extended state space $\mathcal{S}^* = \{1,\dotsc, M\}$ with $M = \sum_{i=1}^N N_i$ is defined as the block matrix
\begin{equation}
    \bm{\Gamma} = 
    \begin{pmatrix}
        \bm{\Gamma}_{11} & \cdots & \bm{\Gamma}_{1N}\\
        \vdots & \ddots & \vdots\\
        \bm{\Gamma}_{N1} & \cdots & \bm{\Gamma}_{NN}\\
    \end{pmatrix},
    \label{eq: Gamma HSMM}
\end{equation}
where the $N_i \times N_i$ diagonal block $\bm{\Gamma}_{ii}, \: i \in \mathcal{S}$ is defined as
\begin{equation}
    \bm{\Gamma}_{ii} = 
    \begin{pmatrix}
        0 & 1-c_i(1) & 0 & \cdots & 0 \\
        0 & 0 & \ddots &  & \vdots \\
        \vdots & \vdots & & & 0 \\
        0 & 0 & \cdots & 0 & 1-c_i(N_i-1) \\
        0 & 0 & \cdots & 0 & 1-c_i(N_i)\\
    \end{pmatrix}
    \label{hsmm_gamma_diag}
\end{equation}
for $N_i \geq 2$, and as $\bm{\Gamma}_{ii} = 1-c_i(1)$ for $N_i=1$, and the $N_i \times N_j$ off-diagonal matrices $\bm{\Gamma}_{ij}, \: i,j \in \mathcal{S}, \: i \neq j$, as
\begin{equation}
    \bm{\Gamma}_{ij} = 
    \begin{pmatrix}
        \omega_{ij} c_i(1) & 0 & \cdots & 0 \\
        \omega_{ij} c_i(2) & 0 & \cdots & 0 \\
        \vdots \\
        \omega_{ij} c_i(N_i) & 0 & \cdots & 0 \\
    \end{pmatrix}.
\end{equation}
In the case $N_j=1$, the zeros disappear. For  $r \in \mathbb{N}$,
\begin{equation}
    c_i(r) = \begin{cases}
        \frac{d_i(r)}{1-F_i(r-1)} \quad & \text{for} \: F_i(r-1) < 1, \\
        1 & \text{for} \: F_i(r-1) = 1,
    \end{cases} 
\end{equation}
where $F_i$ is the cumulative distribution function of $d_i$. The \textit{hazard rates} $c_i$ are the key to this approximation, as they allow the representation of any distribution on the positive integers. They only rely on the parameters of the dwell-time distributions and not on the $\omega_{ij}$; thus, including covariates in the conditional probabilities is straightforward. Transitions within a state aggregate $I_i$, and thus the dwell times in the state aggregate, are governed by the diagonal block $\bm{\Gamma}_{ii}$ while the off-diagonal block matrices $\bm{\Gamma}_{ij}$ contain the probabilities of all possible transitions from state aggregate $I_i$ to state aggregate $I_j$. 
Detailed proofs can be found in \cite{langrock2011hidden}.

Using this representation of the semi-Markov chain as a Markov chain with extended state space, arbitrary dwell-time distributions can be represented. For dwell times larger than $N_i$, the dwell-time distributions of the approximating Markov chain will have a geometric right tail, which in general will not be the case for the dwell-time distributions of the semi-Markov chain that are to be approximated. However, by increasing $N_i$, the approximation can be made arbitrarily accurate. This representation effectively converts an HSMM to a regular HMM with an enlarged state space, allowing for the use of the entire, well-explored HMM methodology like direct numerical maximisation of the likelihood, state-decoding \citep{viterbi1967error}, forecasting and the calculation of pseudo-residuals \citep{zucchini2016hidden}. Moreover, for homogeneous HSMMs, it is convenient to assume stationarity of the approximating chain $\{C_t^*\}$, such that the initial distribution need not be estimated, but can be computed from the transition probability matrix by solving the system of equations
$\bm{\delta} = \bm{\delta} \bm{\Gamma}$, subject to $\sum_{i=1}^M \delta_i = 1$.

Parameters can then be estimated via numerical maximization of the approximate likelihood \citep{langrock2011hidden, macdonald2014numerical} which, for HMMs, can be evaluated recursively using the \textit{forward algorithm} as
\begin{equation}
    \mathcal{L}(\bm{\theta}) = \bm{\delta} \bm{P}(x_1)\bm{\Gamma} \bm{P}(x_2)\bm{\Gamma} \bm{P}(x_3)\bm{\Gamma} \cdot \dotsc \cdot \bm{\Gamma} \bm{P}(x_T)\bm{1},
    \label{eq: likelihood}
\end{equation}
where
$$\bm{P}(x_t) = \text{diag}(\underbrace{f_1(x_t), \dotsc, f_1(x_t)}_{N_i \text{ times}}, \dotsc, \underbrace{f_N(x_t), \dotsc, f_N(x_t)}_{N_N \text{ times}})$$
is a diagonal matrix with state-dependent densities as its entries (each repeated $N_i$ times as each state aggregate is associated with the same state-dependent distribution), and $\bm{1} \in \mathbb{R}^M$ is a column vector of ones. The parameter vector $\bm{\theta}$ then comprises the parameters governing the dwell-time distributions, the conditional transition probabilities and the state-dependent distributions.

\subsection{Inhomogeneous HSMMs}
\label{subsec:Inhomogeneous HSMMs}

The construction of inhomogeneous HSMMs is similar to the homogeneous case. 
Inhomogeneity in such models could concern either the state dwell-time distibutions or the conditional state transition probabilities. For example, the mean dwell time in a state could depend on the time of day the stay in the state is initiated, or the conditional transition probabilities may depend on temperature.

Let again $\{S_k\}$ be the embedded Markov chain, now with conditional t.p.m.\ $\bm{\Omega}^{(k)}$. 
Additionally, let $d_i^{(t)}$, $t = 1, \dotsc, T$, $i \in \mathcal{S}$, be a family of state dwell-time distributions on the positive integers.
In the inhomogeneous case, we have to take some extra care when considering the time scale of the conditional transition probabilities, as the semi-Markov chain and the embedded chain operate on different time scales. For readability, we will now index both the conditional transition probabilities as well as the dwell-time distributions with $t$, i.e.\ on the time-scale of the semi-Markov chain --- which is the natural choice as real data is observed at this scale.
Then, similar to the homogeneous case, we can construct an approximating Markov chain $\{C^*_t\}$ that operates on the extended state space $\mathcal{S^*}$, with structured block t.p.m.\
\begin{equation}
    \bm{\Gamma}^{(t)} = 
    \begin{pmatrix}
        \bm{\Gamma}_{11}^{(t)} & \cdots & \bm{\Gamma}_{1N}^{(t)} \\
        \vdots & \ddots & \vdots \\
        \bm{\Gamma}_{N1}^{(t)} & \cdots & \bm{\Gamma}_{NN}^{(t)} \\
    \end{pmatrix}
    .
    \label{eq: Gamma^t_hsmm}
\end{equation}
The crucial difference lies in the fact that the diagonal block matrices $\bm{\Gamma}_{ii}^{(t)}$ are time-inhomogeneous in this model formulation and are defined as
\begin{equation}
    \bm{\Gamma}_{ii}^{(t)} = 
    \begin{pmatrix}
        0 & 1-c_i^{(t)}(1) & 0 & \cdots & 0 \\
        0 & 0 & 1-c_i^{(t-1)}(2) & \cdots & 0 \\
        0 & 0 & 0 & \ddots & 0 \\
        \vdots & \vdots & & & \vdots \\
        0 & 0 & \cdots & 0 & 1-c_i^{(t-N_i+2)}(N_i-1) \\
        0 & 0 & \cdots & 0 & 1-c_i^{(t-N_i+1)}(N_i)\\
    \end{pmatrix},
    \label{hsmm_gamma_diag_t}
\end{equation}
for $N_i \geq 2$ and $t=1, \dotsc, T$, with $\bm{\Gamma}_{ii}^{(t)} = 1 - c_i^{(t)}(1)$ for $N_i = 1$. The $N_i \times N_j$ off-diagonal matrices $\bm{\Gamma}_{ij}^{(t)}$ also change slightly in their definition to allow for inhomogeneous state dwell times:
\begin{equation}
    \bm{\Gamma}_{ij}^{(t)} = 
    \begin{pmatrix}
        \omega_{ij}^{(t)} c_i^{(t)}(1) & 0 & \cdots & 0 \\
        \omega_{ij}^{(t)} c_i^{(t-1)}(2) & 0 & \cdots & 0 \\
        \vdots \\
        \omega_{ij}^{(t)} c_i^{(t-N_i+1)}(N_i) & 0 & \cdots & 0 \\
    \end{pmatrix}
    .
\end{equation}
In the case $N_j=1$, the zeros disappear. For $r \in \mathbb{N}$ and $t = 1,\dotsc,T$ we define
\begin{equation}
    c_i^{(t)}(r) = \begin{cases}
        \frac{d_i^{(t)}(r)}{1-F_i^{(t)}(r-1)} \quad &\text{for} \: F_i^{(t)}(r-1) < 1, \\
        1 &\text{for} \: F_i^{(t)}(r-1) = 1.
    \end{cases} 
\end{equation}
We see that the hazard functions $c_i^{(t)}$ are now time-dependent, which is the key to the representation of inhomogeneous state dwell-time distributions. We will now show that this model specification yields an approximation to a semi-Markov process with time-varying dwell-time distributions.
\begin{proposition}
\label{prop:omegat}
    Let $\omega_{ij}^{(t) *}$ denote the probability of a transition from state aggregate $I_i$ to $I_j$, i.e.\
    $$\omega_{ij}^{(t) *} = \Pr (C_{t+1}^* \in I_j \mid C_t^* \in I_i, C_{t+1}^* \notin I_i).$$
    We then have
    $$
        \omega_{ij}^{(t)*} = \omega_{ij}^{(t)},
    $$
    for $i \neq j$, $i,j \in \mathcal{S}$.
\end{proposition}
The proof can be found in \ref{A1:proofs}.

Now, we consider the approximation quality regarding the dwell-time distributions, in the inhomogeneous case. Let $d_i^{(t)*}$ likewise denote the p.m.f.\ of the dwell time in state aggregate $I_i$ when the state aggregate is entered at timepoint $t$.
\begin{proposition}
\label{prop: d_i^t}
    For $i \in \mathcal{S}$:
    \begin{equation}
        d_i^{(t)*}(r) \begin{cases}
            = d_i^{(t)}(r)  &\text{for} \; r \leq N_i, \\
            \approx d_i^{(t)}(N_i) \prod_{k=1}^{N_i-r} \bigl(1-c_i^{(t+k)}(N_i)\bigr) \quad &\text{for} \; r > N_i.
        \end{cases}
        \label{eq: representation_di_t}
    \end{equation}
\end{proposition}

While the detailed proof can be found in \ref{A1:proofs}, the core idea remains the same as for homogeneous HSMMs. In this case, however, we want the dwell-time distribution in a state $i$, which is represented by a state aggregate $I_i$ to be the one that is active at the time point the chain transitions into its lowest state $I_i^-$. Due to the time-inhomogeneity, a time shift, as seen above, needs to be implemented such that this is achieved for $r \leq N_i$.

From \eqref{eq: representation_di_t} we see that, in the inhomogeneous case, for $r > N_i$, we cannot expect a geometric tail but only an approximation thereof. Thus, the approximation of the dwell-time distributions for larger dwell times is worse and less controllable than in the homogenous case. Consequently, larger aggregate sizes should be chosen to obtain a close approximation, such that very little mass of the dwell-time distribution is not covered by the aggregate size. This modelling aspect will be investigated more closely in Section \ref{sec:simulation}.

\subsection{Periodically varying dwell-time distributions}
\label{subsec: State dwell times of hidden semi-Markov models}

While we can gain substantially more flexibility by employing HSMMs with inhomogeneous dwell-time distributions that depend on external covariates, such models also lose interpretability. This is because in general, we cannot obtain an overall dwell-time distribution for each state anymore that is unconditional of a specific covariate value. This however would be a desirable summary output, such that the dwell-time distributions can be compared across different states or individuals. 

It is however possible, to deduct such a distribution when restricting to the special case of periodically varying dwell-time distributions,
\begin{equation}
    d_i^{(t)}(r) = d_i^{(t+L)}(r), \quad \text{for all} \; t = 1, \dotsc, T
    \label{eq: periodic dwell-time distributions}
\end{equation}
where $L$ is the integer length of one cycle (for example $L=24$ for hourly sampled data, when the dwell-time distributions vary with the time of day at which a stay is initialised). Furthermore, the conditional transition probabilities are also allowed to vary periodically, leading to
\begin{equation}
    \bm{\Gamma}^{(t)} = \bm{\Gamma}^{(t+L)} \quad \text{for all} \; t = 1, \dotsc, T.
    \label{eq: periodicGamma}
\end{equation}
For such models, certain analytic results can be obtained that allow us to avoid undesired modelling assumptions and arrive at useful summary statistics. First, while for models depending on general covariates, we need to assume a state switch at $t=0$ for feasible estimation of the initial distribution, such a restrictive assumption is not necessary in a periodic setting. By representing the inhomogeneous HSMM by an inhomogeneous HMM with enlarged state space, we can rely on the notion of \textit{periodic stationarity} established by \citet{koslik2023inference}. For Markov chains with periodically varying transition probabilities --- such as $\{C^*_{t}\}$ --- when restricting the dwell times to \eqref{eq: periodic dwell-time distributions}, the unconditional distribution of states can be calculated exactly. For this, consider for a fixed $t$ the thinned Markov chain 
$\{C^*_{t+kL}\}_{k \in \mathbb{N}}$,
which is homogeneous with constant t.p.m.\
$$
\tilde{\boldsymbol{\Gamma}}_t = \boldsymbol{\Gamma}^{(t)} \boldsymbol{\Gamma}^{(t+1)} \ldots \boldsymbol{\Gamma}^{(t+L-1)}.
$$
Provided that this thinned Markov chain is irreducible, it has a unique stationary distribution $\boldsymbol{\delta}^{(t)}$, which is the solution to 
\begin{equation} 
\label{eq: deltat}
\boldsymbol{\delta}^{(t)} = \boldsymbol{\delta}^{(t)} \tilde{\boldsymbol{\Gamma}}_t.
\end{equation}
Therefore we can calculate the stationary distributions on $\mathcal{S}^* = \{1, \dotsc, M\}$ for all time points $t = 1, \dots, L$ within a cycle. We can then easily obtain the periodically stationary state distribution on the original state space $\{1, \dotsc, N\}$ by summing all probabilities within each state aggregate. As the initial distribution, we can then select the one corresponding to the first data point (for details see \cite{koslik2023inference}). 

Moreover, by exploiting \eqref{eq: deltat}, we can proceed to calculate the overall state dwell-time distribution for a given state $i \in \mathcal{S}$ of a periodically inhomogeneous HSMM as a mixture of the time-varying distributions. Due to the approximation, we do in fact calculate the dwell-time distribution for each state aggregate $I_i$ of our approximating Markov chain, however, these two coincide by construction.

\begin{proposition}
    \label{prop: dwell time distr}
    Let $\bm{\Gamma}^{(t)}$ be the 
    t.p.m.\ of an HSMM-approximating HMM on the extended state space $\mathcal{S}^*$, as described in Section \ref{subsec:Inhomogeneous HSMMs}. Also, assume $\bm{\Gamma}^{(t)}$ to vary only periodically due to 
    the state dwell-time distributions $d_i^{(t)}$, and possibly the
    conditional transition probabilities varying periodically, such that \eqref{eq: periodicGamma} is fulfilled. Then, the p.m.f.\ of the overall dwell-time distribution for a stay in state aggregate $I_i$ is
    \begin{equation}
    d_i(r) = \sum_{t=1}^L v_i^{(t)} d_i^{(t)}(r), \quad \text{for} \; r \leq N_i,
    \label{eq: dwell-time distribution phsmm}
    \end{equation}
    with the mixture weights defined as
    $$v_i^{(t)} = \frac{\sum_{l \in I_k: k \neq i} \delta_l^{(t-1)} \gamma_{lI_i^{-}}^{(t-1)}}{\sum_{t=1}^L \sum_{l \in I_k: k \neq i} \delta_l^{(t-1)} \gamma_{lI_i^{-}}^{(t-1)}}, \quad t = 1, \dotsc, L,$$
    and the stationary distribution $\bm{\delta}^{(t)}$ defined as in \eqref{eq: deltat}. Because of their periodic properties, we have that $\bm{\Gamma}^{(0)} = \bm{\Gamma}^{(L)}$ and $\bm{\delta}^{(0)} = \bm{\delta}^{(L)}$.
\end{proposition}
See \ref{A1:proofs} for the proof.

In conclusion, we have established the overall dwell-time distribution implied by a periodically inhomogeneous HSMM as an equivalent tool to the HMM case. This development allows us to effectively compare model-implied overall dwell-time distributions in periodically inhomogeneous settings between HMMs and HSMMs.

\section{Simulation study}
\label{sec:simulation}

Having established inference procedures and important theoretical properties of inhomogeneous HSMMs, some important questions still suggest experimental verification. Therefore, we conduct a simulation study to investigate i) whether the MLE for inhomogeneous HSMMs is consistent, ii) how the choice of state aggregate sizes influences the approximation quality and iii) the severity of the effects of model misspecification.

For the following simulation experiments, all data will be generated from the same model. This model is chosen to be a 3-state HSMM with
state-specific shifted Poisson dwell-time distributions whose means are modelled as a function of the time of day, where the latter is an integer variable ranging from one to 24. The relationship is parametrised as
\begin{equation}
    \lambda_i^{(t)} = \exp \Bigl( \beta_0^{(i)} + \beta_{1}^{(i)} \sin \bigl( \frac{2 \pi t}{24}\bigr) + \beta_{2}^{(i)} \cos \bigl( \frac{2 \pi t}{24}\bigr)\Bigr), \quad i \in \mathcal{S}, \; t \in \{1, \dotsc, 24\},
    \label{eq:lambdat}
\end{equation}
guaranteeing that \eqref{eq: periodic dwell-time distributions} is fulfilled.
The matrix $\bm{\Omega}$ of conditional transition probabilities is taken to be homogeneous and the state-dependent process is bivariate, assuming contemporaneous conditional independence, with gamma-distributed step lengths and von Mises-distributed turning angles (which is a standard choice made for modelling GPS data and thus closely resembles applied settings). All parameters specifying the model for data generation can be found in \ref{A3:parameters}. For the Monte Carlo simulation, all models were fitted in \texttt{R} \citep{R2023} using the numerical optimisation procedure \texttt{nlm} \citep{schnabel1985modular}. To speed up parameter estimation, the forward algorithm was implemented in C++.

\subsection{Consistency of estimators}

Initially, we want to establish consistency, particularly for the estimators of the coefficients in Equation \eqref{eq:lambdat}. We simulate 500 time series of lengths $T = 1000, 2000, 5000, 10000$ from the above-specified model. For each data set, we fit the appropriate inhomogeneous HSMM via the approximation procedure detailed in Section \ref{subsec:Inhomogeneous HSMMs}. Figure \ref{fig:consistency_state1} shows the empirical distribution of the 500 estimates of the coefficients that determine the mean dwell time in the first state. These simulation experiments indicate consistency of the MLE, i.e.\ that it is unbiased with decreasing variance for increasing sample size. Distributions of the MLEs of the coefficients associated with the dwell-time distributions in the other two states can be found in \ref{A2: additional figures}. However, these experiments also indicate that relatively large data sets are required for reliable estimation

We thus conclude that models of comparable complexity should not be applied to small data sets (though estimation accuracy also depends on the separation of the state-dependent distribution, such that no quantitative recommendations can be given regarding the required sample size).

\begin{figure}
    \centering
    \includegraphics[width=1\textwidth]{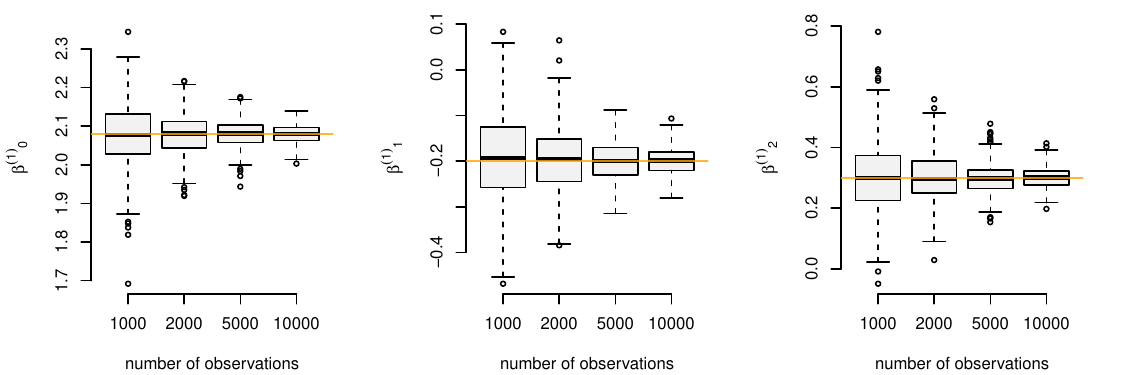}
    \caption{Boxplots of the MLE's coefficients for the mean dwell times of the first state. The true parameters are shown as orange lines.}
    \label{fig:consistency_state1}
\end{figure}

\subsection{Approximation quality in the dwell-time distributions' tails}

From Proposition \ref{prop: d_i^t}, it is evident that, in general, the approximation accuracy regarding the tail of the dwell-time distributions becomes worse and less controllable when incorporating inhomogeneity, compared to the geometric tail exhibited by homogeneous HSMMs. Thus, we want to investigate experimentally which aggregate sizes are necessary to ensure sufficient approximation accuracy. We do this by again simulating 500 data sets, in this case, each of length $T=5000$, and fitting inhomogeneous HSMMs with increasing aggregate sizes. For each state, we consider the 99.5\% quantile of the Poisson dwell-time distribution with its largest mean dwell time and multiply this with a factor ranging from 0.5 to 1.3 by increments of 0.1 (and subsequently round to the nearest larger integer) --- the resulting value is used as the size of the state aggregate. Figure \ref{fig:aggregate_size_state1} shows the distribution of the 500 estimated coefficients that determine the mean dwell time, for factors 0.5, 0.7, 0.9 and 1.3. Factors in the order of 0.5-0.6 --- and hence sizes of the state aggregates that are clearly too small to capture most of the structure of the dwell-time distribution --- result in noticeable biases, especially for the non-intercept coefficients. Increasing the factor to 0.7 already leads to a substantial improvement, and for a factor of 0.9 the approximation becomes virtually exact, as there is no visible difference anymore in the results obtained when using a factor of 0.9 or 1.3. Thus, we can conclude that the not directly controllable tail behaviour of inhomogeneous HSMMs does not constitute a problem in practical applications provided that the aggregate sizes are chosen reasonably large. Guidance for this choice can be obtained by first fitting a homogeneous HMM with geometric dwell times to then calculate e.g.\ the $97.5\%$ quantiles as the size of the associated state aggregate. Figures showing the empirical distributions of the coefficients determining the mean dwell times in the remaining states can be found in Appendix \ref{A2: additional figures}.

\begin{figure}
    \centering
    \includegraphics[width=1\textwidth]{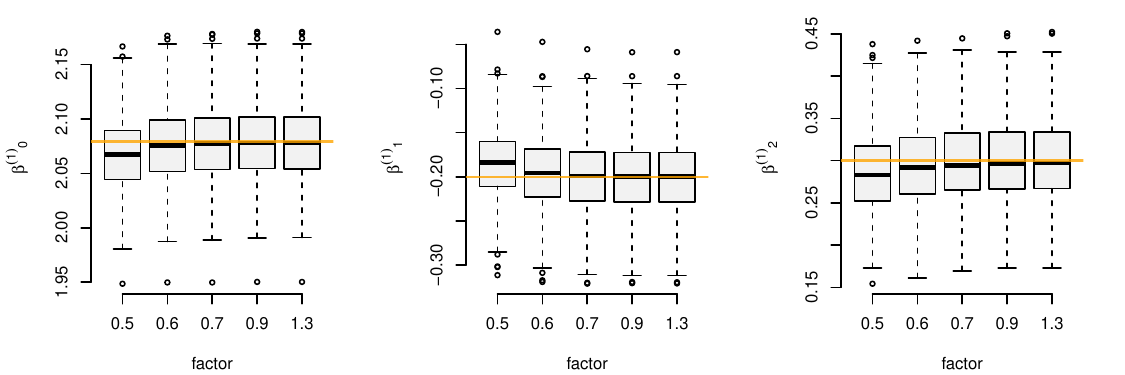}
    \caption{Boxplots of the MLE's coefficients for the mean dwell times of the first state. Each panel shows the distribution of one parameter for increasing aggregate sizes of the approximating HMM. The true parameters are shown as orange lines.}
    \label{fig:aggregate_size_state1}
\end{figure}

\subsection{Model misspecification}

Lastly, we want to assess the severity of model misspecification, specifically when the true data stem from an inhomogeneous HSMM, but an inhomogeneous HMM or a homogeneous HSMM is fitted. To this end, we simulate one time series of length $T = 10^6$ to allow for a simple yet reliable comparison: this way only three distinct fitted models need to be compared, where the large sample size reduces the role of randomness in the comparison (due to the MLE's consistency). Having estimated the three models, we can obtain the most probable, Viterbi-decoded state sequence under each model, to be compared to the true states underlying our simulated data. We obtain decoding accuracies of $96.52\%$, $95.88\%$ and $95.53\%$ for the (correctly specified) periodically inhomogeneous HSMM, the homogeneous HSMM and the periodically inhomogeneous HMM respectively.

\begin{figure}
    \centering
    \includegraphics[width=0.9\textwidth]{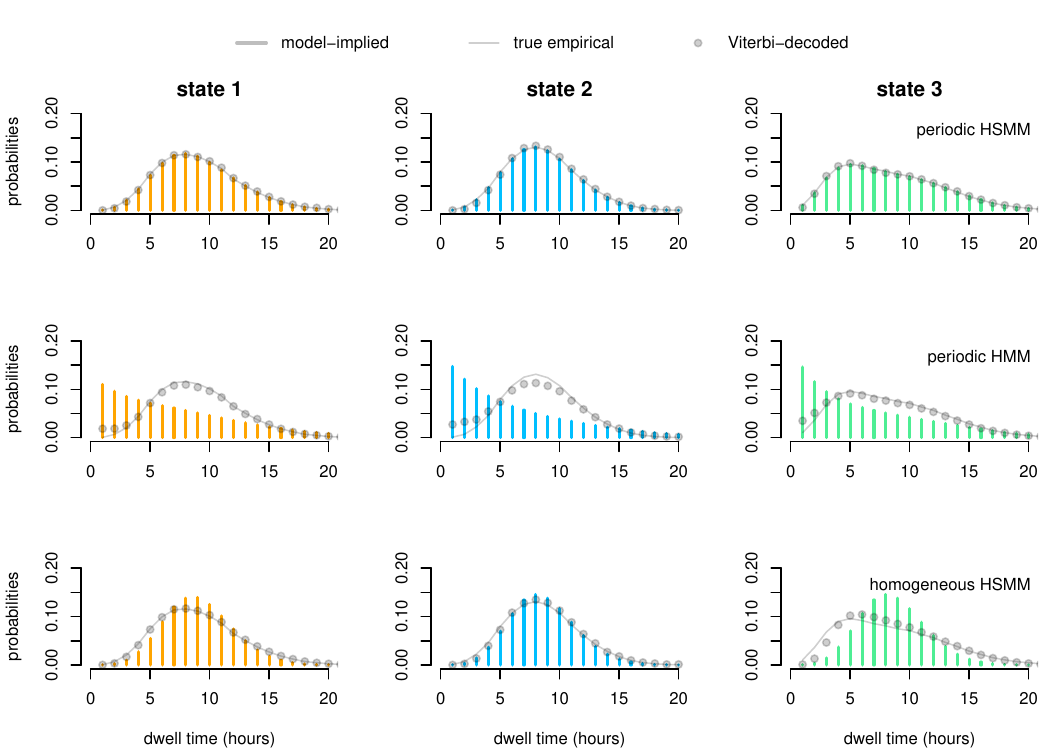}
    \caption{Overall dwell-time distributions in each state for the periodic HSMM, the periodic HMM and the homogeneous HSMM respectively. The model-implied distributions (coloured bars) are complemented by the true empirical distribution (solid grey line) as well as the empirical distribution obtained from a run-length encoding of the Viterbi-decoded states sequences (grey dots).}
    \label{fig:dwell-time_distribution}
\end{figure}

We continue by jointly comparing i) the model-implied overall dwell-time distributions, ii) the true (empirical) dwell-time distributions, and iii) the Viterbi-decoded run lengths.
The latter two can easily be obtained by run-length encoding of the true states and decoded states respectively. The results are shown in Figure \ref{fig:dwell-time_distribution}. For the correct model, all three distributions match very closely (which is to be expected for such a long time series), yet for the two misspecified models, we see a distinct mismatch between the model-implied distribution and the true empirical distribution as well as the decoded empirical distribution. The lack of fit in the state dynamics is strongest for the inhomogeneous HMM. While the homogeneous HSMM can capture most of the structure of the overall dwell-time distribution, it also exhibits a clear lack of fit, mainly by not being able to capture the overdispersion caused by the inhomogeneity in the dwell-time distributions.

Notably, we find that even if a model with misspecified state process dynamics is fitted, 
then the decoded state sequence allows the computation of a close approximation of the true underlying dwell-time distribution. This can be exploited in real analyses as a model-checking tool to uncover misspecification of the state process model.

\section{Application: Muskox movement}
\label{sec:application}

We illustrate the practical potential of fully inhomogeneous HSMMs by applying them to the movement track of an Arctic muskox (\textit{Ovibos moschatus}), collected and analysed by \citet{beumer2020application}. \citet{pohle2022flexible} also used this data set, finding that muskoxen's state-dwell-time distributions were fairly non-geometric. For simplicity, from the 19 movement tracks that were collected, we will only consider the one analysed by \citet{pohle2022flexible} and provided in the \texttt{R} package \texttt{PHSMM}. The data set consists of hourly GPS measurements and contains 6825 observations, including 168 missing values. The GPS measurements were converted into \textit{step lengths} (metres) and \textit{turning angles} (radians) as it is commonly done when analysing horizontal movement data \citep{REFs}. Due to the previous analyses, we have reason to suspect the data to exhibit both temporal variation and non-geometric state dwell times. The diurnal variation in the muskoxen's behaviour is also illustrated in Figure \ref{fig:muskox_boxplot}. 
\begin{figure}
    \centering
    \includegraphics[width=0.7\textwidth]{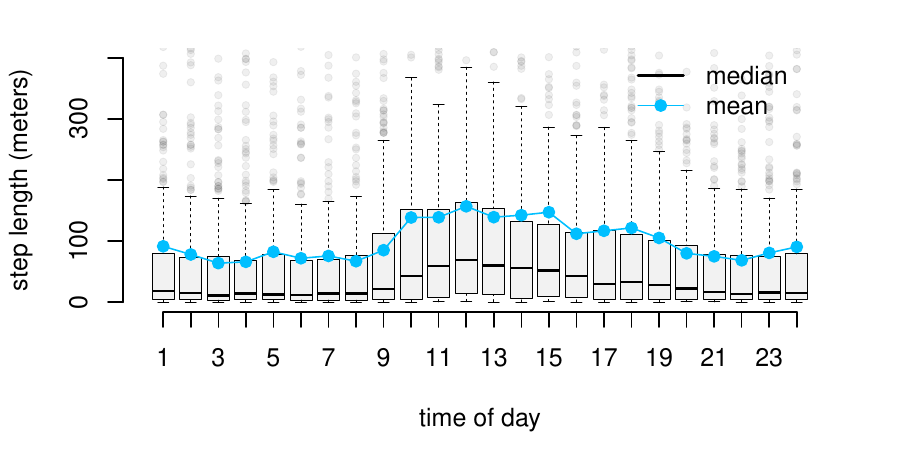}
    \caption{Boxplot of step lengths for each time of the day. The y-axis is limited to 400 meters for visual clarity. Median step length (black) and mean step length (light blue).}
    \label{fig:muskox_boxplot}
\end{figure}
While \citet{beumer2020application} accounted for periodic variation and \citet{pohle2022flexible} flexibly estimated the dwell-time distributions, a combined approach of periodic inhomogeneity in the dwell-time distributions has yet been applied to these data and could reveal novel insights into the animal's behavioural patterns.

As it is not our goal to perform extensive model selection regarding the number of behavioural states, we restrict ourselves to the consideration of models with $N=3$ states. This choice is backed by previous work of \citet{pohle2017selecting}, \citet{beumer2020application} and \citet{pohle2022flexible}, who concluded that 3-state models describe the behaviour of muskoxen adequately. To establish a benchmark, we fit inhomogeneous HMMs with periodically varying transition probabilities $\gamma_{ij}^{(t)}$ based on linear predictors of the form
$$
\eta_{ij}^{(t)} = \beta_0^{(ij)} + \sum_{k=1}^{K_{\bm{\Gamma}}} \beta_{1k}^{(ij)} \sin \Bigl( \frac{2 \pi k t}{24} \Bigr) + \sum_{k=1}^{K_{\bm{\Gamma}}} \beta_{2k}^{(ij)} \cos \Bigl( \frac{2 \pi k t}{24} \Bigr), \quad i \neq j,
$$
where each row of the t.p.m.\ is then computed via the inverse multinomial logistic link. We fit models of degree $K_{\bm{\Gamma}} = 0, 1, 2$, where $0$ refers to homogeneous transition probabilities. For the HSMMs, we use shifted negative binomial dwell-time distributions, to increase flexibility compared to both geometric but also shifted Poisson distributions by estimating an additional dispersion parameter. Inhomogeneity of these state-specific dwell-time distributions can then be realised by modelling the mean $\mu$ and possibly the dispersion parameter $\phi$ (where the variance is $\mu + \mu^2 \phi$) by trigonometric functions of degrees $K_{\mu}$ and $K_{\phi}$, respectively, similar to Equation \eqref{eq:lambdat}. Furthermore, the conditional transition probabilities $\omega_{ij}^{(t)}$ can also be modelled by trigonometric functions with degree $K_{\bm{\Omega}}$. This allows for a wide range of candidate models providing great flexibility. However, to avoid instability, we do not allow degrees higher than one for the dispersion parameters and the conditional transition probabilities. The state-dependent process is modelled assuming contemporaneous conditional independence, with state-dependent gamma and von Mises distributions for step lengths and turning angles, respectively.

A total of nine candidate models were fitted in \texttt{R} \citep{R2023} using the parallelised numerical optimisation procedure \texttt{optimParallel} \citep{gerber2019optimparallel} and an implementation of the forward algorithm in \texttt{C++} to speed up the estimation. The resulting maximum log-likelihoods and information criteria are shown in Table \ref{tab:IC_muskox}.

\vspace{0.5cm}
\begin{table}[H]
\centering
\begin{tabular}{lcccccccc}
\toprule
Model type & No. & $K_{\bm{\Gamma}}$ & $K_{\mu}$ & $K_{\phi}$ & $K_{\bm{\Omega}}$ & $\bm{\ell}$ & $AIC$ & $BIC$ \\
\midrule
\multirow{3}{*}{HMMs} & 1 & 0 & - & - & - & -44790.5 & 89611.0 & 89713.5 \\
{} & 2 & 1 & - & - & - & -44750.0 & 89554.0 & 89738.4 \\ 
{} & 3 & 2 & - & - & - & -44728.1 & 89534.3 & 89800.6 \\  
\midrule
\multirow{6}{*}{HSMMs} & 4 & - & 0 & 0 & 0 & -44726.5 & 89489.1 & 89612.0 \\
{} & 5 & - & 1 & 0 & 0 & -44686.1 & 89420.2 & \textbf{89584.1} \\ 
{} & 6 & - & 2 & 0 & 0 & -44679.9 & 89419.7 & 89624.6 \\ 
{} & 7 & - & 1 & 1 & 0 & -44666.4 & 89392.7 & 89597.6 \\
{} & 8 & - & 1 & 0 & 1 & -44679.9 & 89419.8 & 89624.6 \\
{} & 9 & - & 1 & 1 & 1 & -44652.4 & \textbf{89376.8} & 89622.7 \\
\bottomrule
\end{tabular}
\caption{Comparison of model performances for the muskox data according to AIC and BIC. The lowest values are highlighted in boldface.}
\label{tab:IC_muskox}
\end{table}

We find that HSMMs improve the model fit drastically compared to HMMs and that within the class of HSMMs, models that account for periodic variation clearly outperform homogeneous models. While the substantial flexibility provided by model 9 results in the best AIC value, caution has to be taken when performing model selection regarding the dwell-time distributions' parameters: Such a model formulation results in a state-specific GLMSS \citep{gamlss2005} performed on \textit{latent structures}, in addition to the inhomogeneity in the conditional transition probabilities. Moreover, it should be stated that estimation of the time-varying dwell-time distributions exhibits varying precision at different times of the day, as some of the time points provide little information if the animal does not transition into a certain state at one of these time points. From the simulation experiments, it also became evident that the coefficients determining the mean dwell time require a considerable amount of data to be estimated precisely (see Figure \ref{fig:consistency_state1}). Thus, being conservative here, we examine the results of model 5 more thoroughly while comparing it to its baseline competitors, i.e.\ the inhomogeneous HMM with $K_{\bm{\Gamma}}=1$ (model 2) and the homogeneous HSMM (model 4). 

The marginal distributions can be obtained by averaging the periodically stationary state probabilities, and we observe a satisfactory fit (see Figure \ref{fig:marginal} in \ref{A2: additional figures}). Inspecting the estimated component distributions, we can interpret the three states as proxies for \textit{resting}, \textit{foraging} and \textit{travelling} behaviour.

Focussing on the state process dynamics, in such a 3-state scenario the interpretation of nine time-varying transition probabilities as functions of the time of day can be tedious; thus, we can consider the periodically stationary distribution (see Equation \eqref{eq: deltat}) as a more interpretable summary statistic. For HSMMs, we easily obtain it from the approximating HMM by summing the periodically stationary probabilities within each state aggregate at each time point of the day. These periodically stationary state probabilities of the inhomogeneous HSMM (model 5) and inhomogeneous HMM (model 2) are shown in Figure \ref{fig:stationary_p_muskox}. For both models, we can observe that resting becomes more likely at nighttime and that the probability of occupying either of the two active states increases during the day.
\begin{figure}[H]
    \centering
    \includegraphics[width=0.9\textwidth]{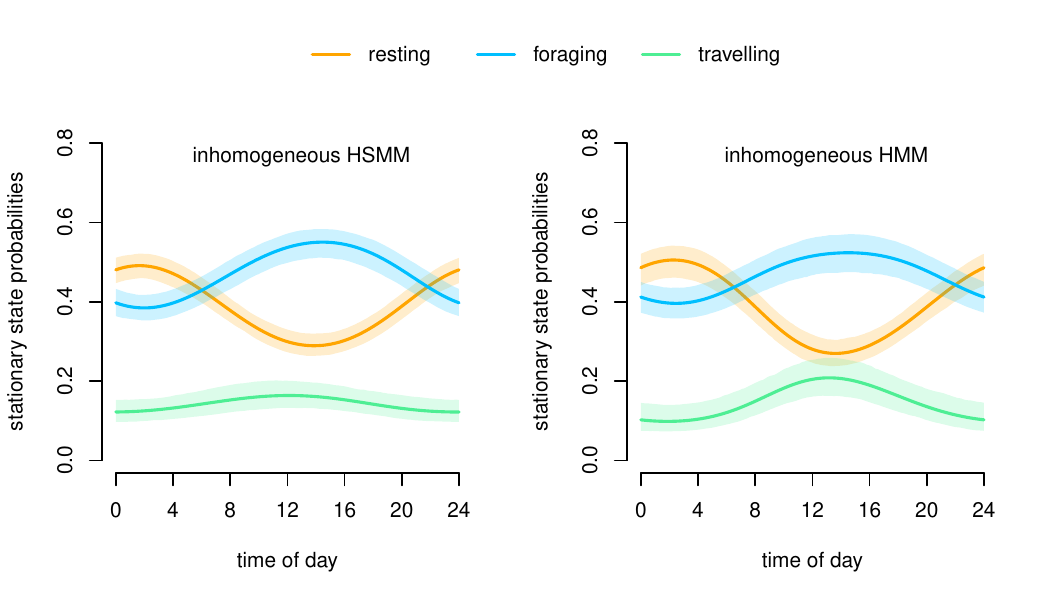}
    \caption{Periodically stationary distribution obtained from the inhomogeneous HSMM (model 2) and inhomogeneous HMM (model 2). Pointwise 95\% confidence intervals were obtained via Monte Carlo simulation from the approximate normal distribution of the maximum likelihood estimator.}
    \label{fig:stationary_p_muskox}
\end{figure}
 While the inhomogeneous HSMM shows less periodic variability in the travelling probability compared to the inhomogeneous HMM, it nonetheless exhibits a considerable amount of for the resting and foraging states. As such, it allows for detailed behavioural inference regarding the effects of diurnal variation, where the homogeneous HSMM fails to capture any periodic effects.
 


We proceed to investigate the time-varying dwell-time distributions obtained from the inhomogeneous HSMM. Figure \ref{fig:time_varying_distr_heat} shows the mean dwell time (which is explicitly modelled), and in addition, the implied shifted negative binomial distributions as a function of the time of day, for each state. We find a substantial variation in the mean dwell time over a day, where the resting state has particularly long expected dwell times in the late evening hours, dwell times in the foraging state are longest when beginning a stay at noon and the travelling state shows its longest dwell times in the morning. The isolated interpretation of these distributions is however limited, as they only indicate the conditional probability of a certain stay length, while not accounting for the probability that the animal actually begins a stay at the time point of interest. Nevertheless, the substantial variation in mean dwell times underlines the superiority of the inhomogeneous HSMM over its homogeneous counterpart, as the latter cannot account for such variation at all.

\begin{figure}[H]
    \centering
    \includegraphics[width=1\textwidth]{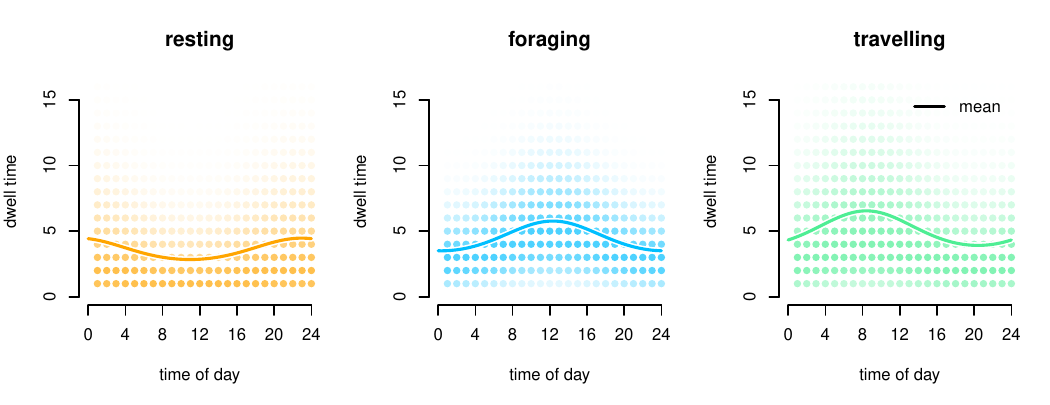}
    \caption{Time-varying dwell-time distributions as a function of the time-of-day. The distribution is visualised by colour-coding each dwell time by its associated probability. For better interpretability, the estimated time-varying mean is also shown.}
    \label{fig:time_varying_distr_heat}
\end{figure}

To obtain a simpler, more interpretable picture that incorporates the probability that a stay is actually started, we use the \textit{overall dwell-time distribution} implied by the HSMM as developed in Section \ref{sec:Methodology} as a complementary inference tool. Figure \ref{fig:overall_ddwell_muskox} shows this overall distribution, not only for the inhomogeneous HSMM but also for its homogeneous counterpart and the inhomogeneous HMM, as obtained by its respective overall dwell-time distribution \citep{koslik2023inference}. We complement these distributions with the empirical distributions obtained from a run-length encoding of each model's Viterbi-decoded state sequence, as explained in Section \ref{sec:simulation}. We find a stark difference when comparing the three models. In particular, it becomes evident why the HSMMs are superior to the HMM, as both models capture the overall shape of the dwell-time distributions notably better compared to the HMM. While, in general, the dwell-time distributions implied by periodically inhomogeneous HMMs can also deviate substantially from a geometric distribution as shown by \citet{koslik2023inference}, and Figure \ref{fig:stationary_p_muskox} indicates a considerable amount of periodic variation, in this case, it does not seem sufficiently strong to allow the overall dwell-time distribution to differ much from a geometric shape. Particularly, the muskoxen appears to have additional, unobserved reasons to exhibit dwell times that can neither be captured well by a geometric distribution nor by periodic variation alone. Overall the two HSMMs dwell-time distributions do not differ substantially and both show a minor lack of fit with respect to the travelling mode. It seems as if the model-implied mixture of negative binomial distributions is not flexible enough to capture the slightly bimodal shape present in the empirical distribution and in the flexible version estimated by \citet{pohle2022flexible}.
\begin{figure}
    \centering
    \includegraphics[width=1\textwidth]{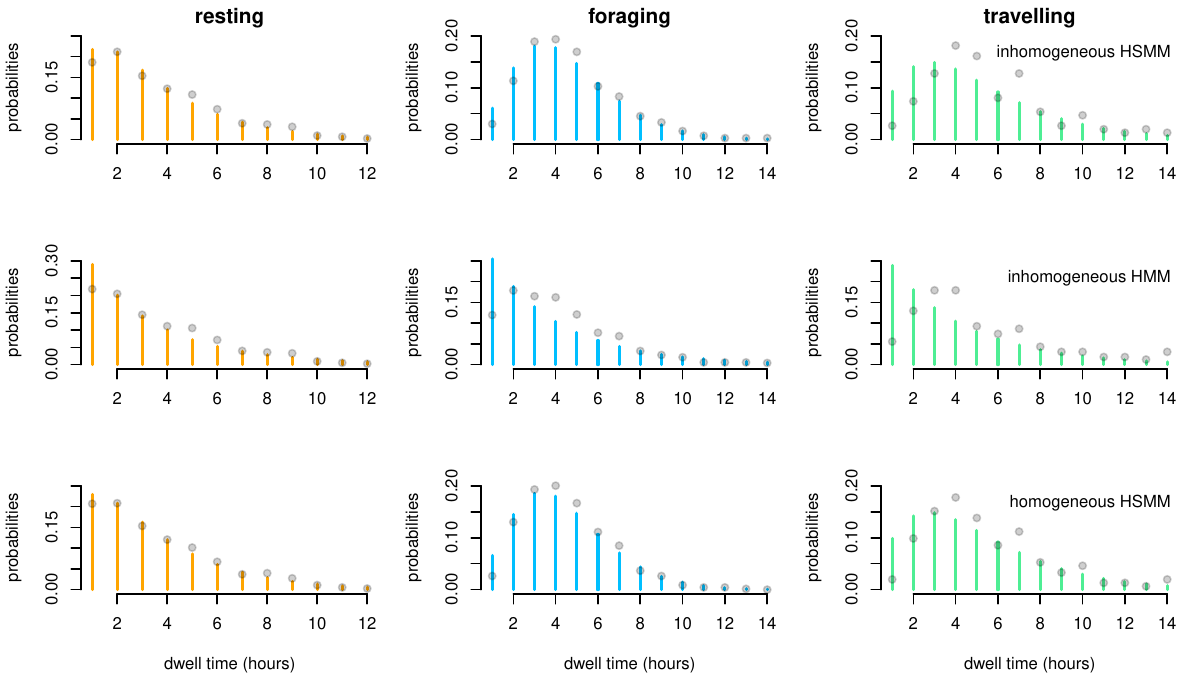}
    \caption{Overall dwell-time distributions of the inhomogeneous HSMM and its two baseline competitors. The coloured bars show the model-implied overall dwell-time distribution, while the dots are the empirical dwell-time distributions obtained from the respective models' Viterbi-decoded state sequence.}
    \label{fig:overall_ddwell_muskox}
\end{figure}
Nonetheless, from an inference perspective, it is highly beneficial to allow for periodic variation in the dwell-time distributions in this application, as we can still learn about the periodic variation in behavioural modes while using a parsimonious model that incorporates the dwell time structure very well. Thus, the novel model class of fully inhomogeneous HSMMs does offer additional flexibility in modelling real data.

\section{Discussion}
\label{sec:discussion}

In several areas of empirical research, HSMMs are already well-established tools for modelling processes with latent components. They extend regular HMMs and are suited in particular in scenarios when the latter do not provide sufficient flexibility regarding the modelling of the state process dynamics. However, while inhomogeneity can nowadays easily be incorporated into the state process of HMMs, a comparable framework had been missing for HSMMs. In particular, the inference procedure established by \citet{langrock2011hidden} only allowed for covariate effects on the conditional transition probabilities, given that the state is left, while leaving the state dwell-time distributions homogeneous.

We successfully addressed this limitation by expanding the HSMM formulation to accommodate inhomogeneity in the state dwell times. As a result, we achieved the estimation of \textit{fully inhomogeneous HSMMs} through direct numerical approximate maximum likelihood estimation. Thereby, we overcame previous hurdles and broadened the scope of practical applications for these models while providing a very general estimation approach that allows for wide flexibility regarding the inclusion of inhomogeneity in the dwell-time distributions, as well as in the conditional transition probabilities. In addition to establishing the technical details necessary for inference in such models, we derived important properties for the special case of periodically varying dwell-time distributions, extending the results of \citet{koslik2023inference} and thereby allowing for better model interpretability. Complementary to these 
theoretical properties
, we conducted a simulation study to investigate the consistency of the MLE, to shed some light on how to choose suitable sizes of the state aggregates for the approximate inference procedure and to explore the effects of model misspecification.

To showcase the practical relevance of the novel model class, we applied it to a case study of muskox movement. The results suggest that inhomogeneous HSMMs can be a considerably better choice when the data is characterised by both non-monotonous state dwell times and considerable periodic variation.
In such scenarios, we argue that the avoidance of unrealistic assumptions and the resulting improvement in the model fit can justify the additional computational overhead of HSMMs. 

\section*{Acknowledgements}
The author gratefully acknowledges Roland Langrock for helpful comments on an earlier version of this manuscript.

\bibliographystyle{elsarticle-harv} 
\bibliography{cas-refs}

\newpage
\appendix

\section{Proofs}
\label{A1:proofs}

\begin{proof}[Proof of Proposition \ref{prop:omegat}]
\begin{align*}
    \omega_{ij}^{(t)*} &= \frac{\Pr (C_{t+1}^* \in I_j , C_t^* \in I_i)}{\Pr (C_t^* \in I_i, C_{t+1}^* \notin I_i)} \\
    &= \frac{\sum_{r=1}^{N_i} \omega_{ij}^{(t)} c_i^{(t-r+1)}(r)}{\sum_{j=1}^N \sum_{r=1}^{N_i} \omega_{ij}^{(t)} c_i^{(t-r+1)}(r)} \\
    &= \frac{\omega_{ij}^{(t)} \sum_{r=1}^{N_i} c_i^{(t-r+1)}(r)}{\sum_{r=1}^{N_i} c_i^{(t-r+1)}(r)} = \omega_{ij}^{(t)}
\end{align*}
\end{proof}

\begin{lemma}
    Let $i \in \{1, \dotsc, N\}$ and $n \in \mathbb{N}$. Then
    \begin{equation}
        \prod_{k=1}^n \left(1-c_i(k)\right) = 1 - F_i(n).
        \label{eq: cum_distr_HSMM}
    \end{equation}

    \begin{proof}
    By induction. For the base case it holds $1-c_i(1) = 1-d_i(1) = 1-F_i(1)$. Now assume that \eqref{eq: cum_distr_HSMM} holds for some $n \in \mathbb{N}$. If $F_i(n) < 1$ then
    \begin{align*}
        \prod_{k=1}^{n+1} \left(1-c_i(k)\right) &= \left( \prod_{k=1}^{n} \left(1-c_i(k)\right)\right) \left(1 - c_i(n+1)\right)\\
        &= \left( 1 - F_i(n)\right) \left( 1 - \frac{d_i(n+1)}{1 - F_i(n)} \right)\\
        &= 1 - \left( F_i(n) + d_i(n+1)\right) = 1 - F_i(n+1).
    \end{align*}
    If $F_i(n)=1$ then $F_i(n+1)=1$ and
    $$\prod_{k=1}^{n+1} \left( 1 - c_i(k) \right) = \left( 1 - F_i(n)\right) \left( 1 - c_i(n+1) \right) = 0 = 1 - F_i(n+1).$$
\end{proof}
\end{lemma}

\begin{proof}[Proof of Proposition \ref{prop: d_i^t}]
    The case $N_i = 1$ is trivial, so we consider the case $N_i \geq 2$. Since every sojourn in the state aggregate $I_i$ starts in state $I_i^-$, and taking into account the special structure of $\bm{\Gamma}^{(t)}$, it follows that
    \begin{align*}
        d_i^{(t)*}(1) &= Pr\left(C_{t+1}^* \notin I_i \mid C_t^* \in I_i, C^*_{t-1} \notin I_i \right)\\
        &= \sum_{l \in \mathcal{S} : l \neq i} \omega_{il}^{(t)} c_i(1)^{(t)} = c_i(1)^{(t)} = d_i^{(t)}(1).
    \end{align*}
    
    We now consider $d_i^{(t)*}(r)$ for $2 \leq r \leq N_i$. The structure of the matrices $\bm{\Gamma}^{(t)}, \bm{\Gamma}^{(t+1)}, \dotsc, \bm{\Gamma}^{(t+r)}$ is such that the dwell time in state aggregate $I_i$, starting in time point $t$, is of length $r$ if and only if the state sequence successively runs through the states $I_i^-$ in $t$, $I_i^-+1$ in $t+1$ up to $I_i^-+r-1$ in $t+r-1$ and the switches from state $I_i^- +r-1$ to a different state aggregate in $t+r$.
    If $F_i^{(t)}(r-1) < 1$ then from the special structure of $\bm{\Gamma}^{(t)}$ we see that \eqref{eq: cum_distr_HSMM} again applies
    \begin{align*}
        d_i^*(r) &= \prod_{k=0}^{r-2} \left( 1-c_i^{((t+k)-k}(k+1)) \right) \sum_{l \in \mathcal{S} \setminus i} \omega^{(t+r-1)}_{il} c_i^{((t+r-1)-(t+r-1))}(r)\\
        &= \prod_{k=1}^{r-1} \left( 1-c_i^{(t)}(k)) \right) \sum_{l \in \mathcal{S} \setminus i} \omega^{(t+r-1)}_{il} c_i^{(t)}(r)\\
        &= \left( 1 - F_i^{(t)}(r-1)\right) \frac{d_i^{(t)}(r)}{1 - F_i^{(t)}(r-1)} = d_i^{(t)}(r),
    \end{align*}
    as $t$ is fixed. If $F_i^{(t)}(r-1) = 1$ then $d_i^{(t)*}(r) = 0 = d_i^{(t)}(r)$.
    
    Again, we consider the case $r > N_i$. The dwell time in state aggregate $I_i$, starting in time point $t$, is of length $r > N_i$ if and only if the state sequence successively runs through the states $I_i^-$ in $t$, $I_i^-+1$ in $t+1$ up to $I_i^-+N_i-1 = I_i^+$ in $t+N_i-1$ and then remains in state $I_i^+$ from $t+N_i-1$ up to $t+r-1$, and finally switches to a different state aggregate. If $F_i^{(t)}(N_i-1)<1$, then
    \begin{align*}
        d_i^{(t)*}(r) &= \prod_{k=1}^{N_i-1} \left( 1 - c_i^{(t)}(k)\right) \prod_{k=0}^{r-N_i-1} \left(1-c_i^{(t+k)}(N_i)\right) \sum_{l \in \mathcal{S} \setminus i} \omega_{li}^{(t+r-1)} c_i^{(t+r-N_i+1)}(N_i)\\
        &= \left(1-F_i^{(t)}(N_i-1)\right) c_i^{(t+r-N_i+1)}(N_i) \prod_{k=0}^{r-N_i-1} \left(1-c_i^{(t+k)}(N_i)\right)\\
        &\approx d_i^{(t)}(N_i) \prod_{k=0}^{r-N_i-1} \left(1-c_i^{(t+k)}(N_i)\right)
    \end{align*}
    for $N_i$ large enough.
\end{proof}

\begin{proof}[Proof of Proposition \ref{prop: dwell time distr}]
    We consider a certain time point $t \in \{1, \dotsc, L\}$ to be a realisation of the random variable $\tau$. Then we start by noting that in the HSMM setting, for $r \leq N_i$
    \begin{align*}
        d_i^{(t)}(r) = \Pr(C_{\tau+r}^* \notin I_i, C_{\tau+r-1}^* \in I_i, \dotsc, C_{\tau+1}^* \in I_i \mid C_{\tau}^* \in I_i, C_{\tau-1}^* \notin I_i, \tau = t).
    \end{align*}
    The probability we are interested in as the overall dwell-time distribution is 
    \begin{align*}
        \Pr(C_{\bigcdot+r}^* \notin I_i, C_{\bigcdot+r-1}^* \in I_i, \dotsc, C_{\bigcdot+1}^* \in I_i \mid C_{\bigcdot}^* \in I_i, C_{\bigcdot-1}^* \notin I_i).
    \end{align*}
    Again, the dot notation is used to emphasise that we condition on the event that the transition from a different state aggregate to the aggregate $I_i$ has happened at some time point in the cycle. We can obtain the expression from above as a mixture of the time-varying dwell-time distributions. Choose the mixture weights as
    \begin{align*}
        v_i^{(t)} &= \Pr(C^*_{\tau} \in I_i, C^*_{\tau-1} \notin I_i, \tau = t \mid C^*_{\bigcdot} \in I_i, C^*_{\bigcdot-1} \notin I_i) \\
        &= \frac{\Pr(C^*_{\tau} \in I_i, C^*_{\tau-1} \notin I_i, \tau = t)}{\Pr(C^*_{\bigcdot} \in I_i, C^*_{\bigcdot-1} \notin I_i)},
    \end{align*}
    where the equality is justified by the fact that the former event is contained in the latter. Then, we can obtain the overall dwell-time distribution as
    \begin{align*}
    d_i(r) &= \sum_{t=1}^L d_i^{(t)}(r) \cdot v_i^{(t)}\\
    &= \sum_{t=1}^L \Pr(C_{\tau+r}^* \notin I_i, C_{\tau+r-1}^* \in I_i, \dotsc, C_{\tau+1}^* \in I_i \mid C_{\tau}^* \in I_i, C_{\tau-1}^* \notin I_i, \tau = t)\\
    & \qquad \cdot \frac{\Pr(C^*_{\tau} \in I_i, C^*_{\tau-1} \notin I_i, \tau = t)}{\Pr(C^*_{\bigcdot} \in I_i, C^*_{\bigcdot-1} \notin I_i)}\\
    &= \frac{\sum_{t=1}^L \Pr(C_{\tau+r}^* \notin I_i, C_{\tau+r-1}^* \in I_i, \dotsc, C_{\tau+1}^* \in I_i, C_{\tau}^* \in I_i, C_{\tau-1}^* \notin I_i, \tau = t)}{\Pr(C^*_{\bigcdot} \in I_i, C^*_{\bigcdot-1} \notin I_i)} \\
    &= \frac{ \Pr(C_{\bigcdot+r}^* \notin I_i, C_{\bigcdot+r-1}^* \in I_i, \dotsc, C_{\bigcdot+1}^* \in I_i, C_{\bigcdot}^* \in I_i, C_{\bigcdot-1}^* \notin I_i)}{\Pr(C^*_{\bigcdot} \in I_i, C^*_{\bigcdot-1} \notin I_i)}\\
    &= \Pr(C_{\bigcdot+r}^* \notin I_i, C_{\bigcdot+r-1}^* \in I_i, \dotsc, C_{\bigcdot+1}^* \in I_i \mid C_{\bigcdot}^* \in I_i, C_{\bigcdot-1}^* \notin I_i).
    \end{align*}
    We calculate the mixture weights explicitly by splitting the numerator and denominator. For the numerator, we need to consider all possible paths from all states in different state aggregates $I_k$, $k \neq i$ to state aggregate $I_i$, as a stay in state aggregate $I_i$ only begins when the chain has previously been in a different state aggregate. As $I_i$ can only be entered in its lowest state $I_i^-$ we get
    \begin{align*}
        &\Pr(C^*_{\tau} \in I_i, C^*_{\tau-1} \notin I_i, \tau = t) = \sum_{l \in I_k, k \neq i} \delta_{l}^{(t-1)} \gamma_{l I_i^-}^{(t-1)} \frac{1}{L}.
    \end{align*}
    For the denominator, we obtain
    \begin{align*}
        \Pr(C^*_{\bigcdot} \in I_i, C^*_{\bigcdot-1} \notin I_i) &= \sum_{t=1}^L \Pr(C^*_{\tau} \in I_i, C^*_{\tau-1} \notin I_i, \tau = t) \\
        &= \frac{1}{L} \sum_{t=1}^L \sum_{l \in I_k, k \neq i} \delta_{l}^{(t-1)} \gamma_{l I_i^-}^{(t-1)}.
    \end{align*}
    And therefore, by combining the numerator and denominator, we get:
    \begin{align*}
        v_i^{(t)} = \frac{\sum_{l \in I_k, k \neq i} \delta_l^{(t-1)} \gamma_{lI_i^{-}}^{(t-1)}}{\sum_{t=1}^L \sum_{l \in I_k, k \neq i} \delta_l^{(t-1)} \gamma_{l I_i^{-}}^{(t-1)}}.
    \end{align*}
\end{proof}

\newpage

\section{Additional figures}
\label{A2: additional figures}

\begin{figure}[H]
    \centering
    \includegraphics[width=1\textwidth]{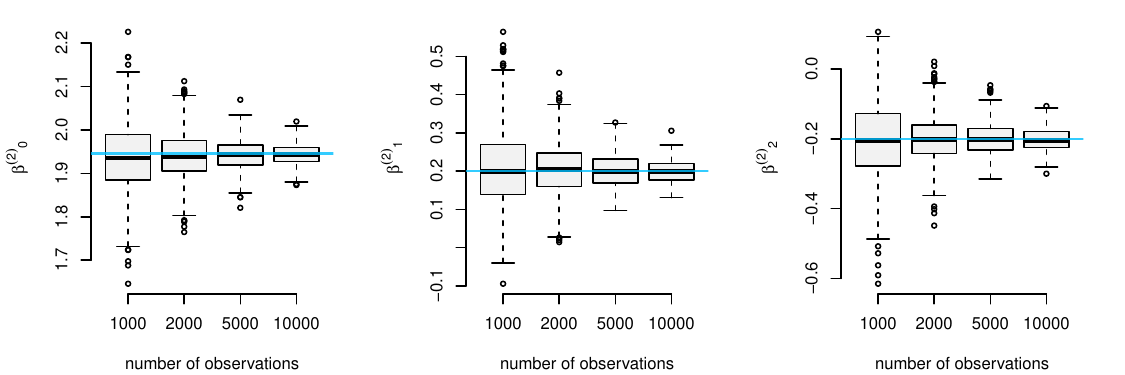}
    \caption{Boxplots of the MLE's coefficients for the mean dwell times of the second state. Each panel shows the distribution of one parameter for increasing sample size. The true parameters are shown as blue lines.}
    \label{fig:consistency_state2}
\end{figure}

\begin{figure}[H]
    \centering
    \includegraphics[width=1\textwidth]{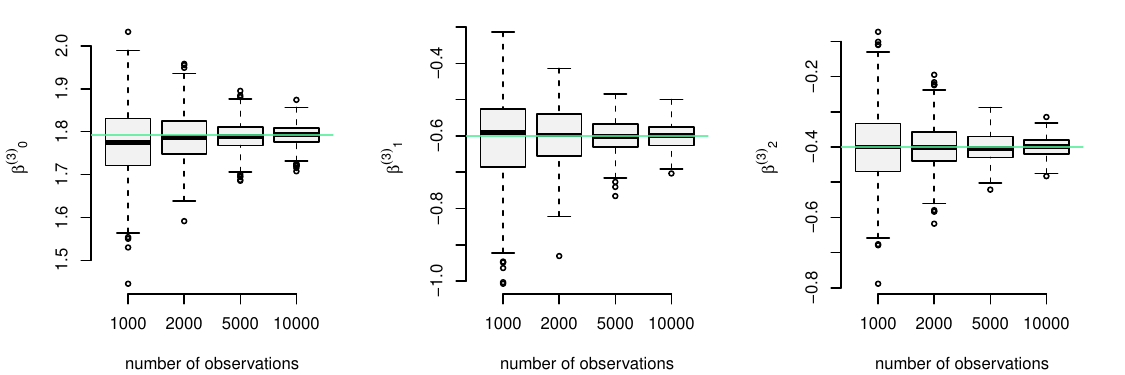}
    \caption{Boxplots of the MLE's coefficients for the mean dwell times of the third state. Each panel shows the distribution of one parameter for increasing sample size. The true parameters are shown as green lines.}
    \label{fig:consistency_state3}
\end{figure}

\begin{figure}[H]
    \centering
    \includegraphics[width=1\textwidth]{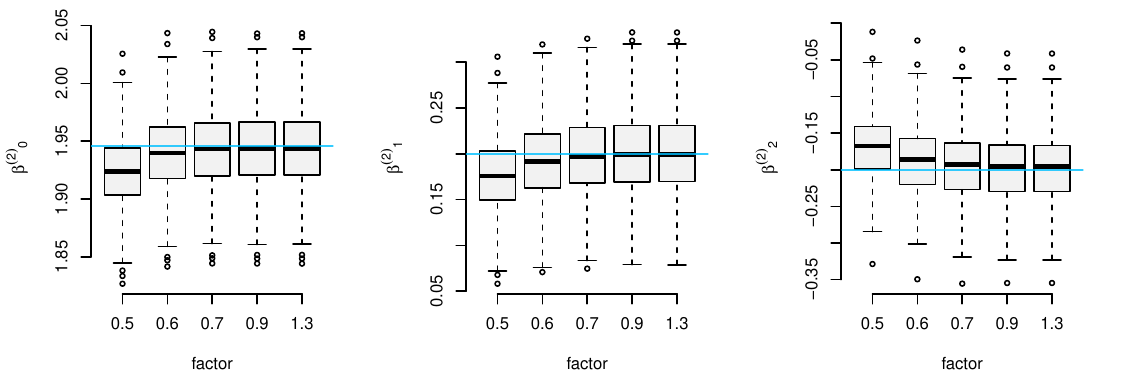}
    \caption{Boxplots of the MLE's coefficients for the mean dwell times of the third state. Each panel shows the distribution of one parameter for increasing aggregate sizes of the approximating HMM. The true parameters are shown as blue lines.}
    \label{fig:aggregate_size_state2}
\end{figure}

\begin{figure}[H]
    \centering
    \includegraphics[width=1\textwidth]{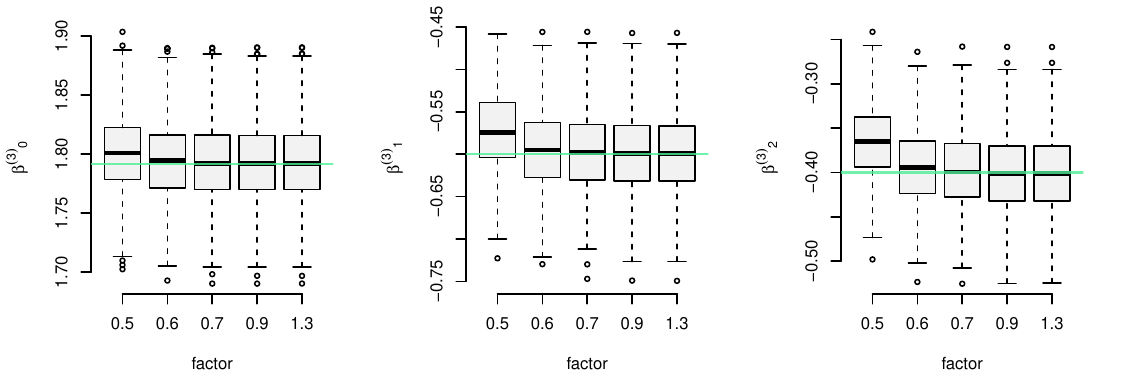}
    \caption{Boxplots of the MLE's coefficients for the mean dwell times of the third state. Each panel shows the distribution of one parameter for increasing aggregate sizes of the approximating HMM. The true parameters are shown as green lines.}
    \label{fig:aggregate_size_state3}
\end{figure}

\begin{figure}[H]
    \centering
    \includegraphics[width=1\textwidth]{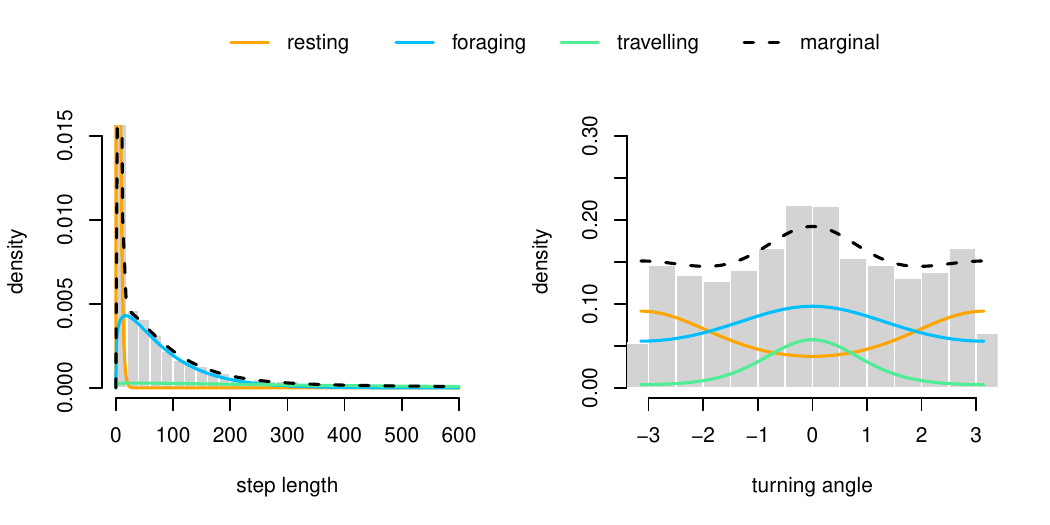}
    \caption{Weighted state-dependent and marginal distributions of the muskox's step lengths and turning angles obtained from the estimated inhomogeneous HSMM}
    \label{fig:marginal}
\end{figure}

\newpage

\subsection{Parameters used for simulation}
\label{A3:parameters}

The parameters used to generate data from an inhomogeneous HSMM in Section \ref{sec:simulation} are:

$$\bm{\beta}^{(1)} = (\log(8), -0.2, 0.3)$$
$$\bm{\beta}^{(2)} = (\log(7), 0.2, -0.2)$$
$$\bm{\beta}^{(2)} = (\log(6), -0.6, 0.4)$$

$$\bm{\Omega} = \begin{pmatrix}
    0 & 0.7 & 0.3 \\
    0.2 & 0 & 0.8 \\
    0.5 & 0.5 & 0
\end{pmatrix}$$

$$\bm{\mu}_{step} = (20, 200, 800)$$
$$\bm{\sigma}_{step} = (20, 150, 500)$$
$$\bm{\kappa}_{angle} = (0.2, 1, 2.5)$$

\end{document}